\documentclass[a4paper,fleqn]{cas-dc2}
\usepackage{cas-dfrws} 
\renewcommand{\dfrwsconference}{Digital Forensics Research Conference USA
(DFRWS USA), 2023}

\usepackage{stix}

\usepackage{graphicx}
\usepackage{array, multirow}
\usepackage{enumitem}
\usepackage[utf8]{inputenc}
\usepackage[dvipsnames]{xcolor}
\usepackage[colorlinks]{hyperref}
\hypersetup{
    linkcolor=CornflowerBlue,
    citecolor=CornflowerBlue,
    urlcolor=cyan
}
\usepackage{algorithm}
\usepackage{algpseudocode}
\usepackage{amsmath}
\usepackage{amsthm}
\usepackage{amsfonts}
\usepackage{mathtools}

\newtheorem{definition}{Definition}
\newtheorem{proposition}[definition]{Proposition}
\newtheorem{theorem}[definition]{Theorem}

\usepackage[authoryear,longnamesfirst]{natbib}
 \def\hyph{-\penalty0\hskip0pt\relax}
\begin{document}

\shortauthors{Ottmann et~al.}

\author[1]{Jenny Ottmann}[orcid=0000-0003-1090-0566]
\cormark[1]
\ead{jenny.ottmann@fau.de}
\credit{Conceptualization, Methodology, Investigation, Software, Writing - Original Draft, Writing - Review and Editing}

\author[1]{{\"U}same Cengiz}[orcid=0009-0004-4092-7668]
\ead{uesame.cengiz@fau.de}
\credit{Methodology, Investigation, Writing - Review and Editing}

\author[2]{Frank Breitinger}[orcid=0000-0001-5261-4600]
\ead{frank.breitinger@unil.ch}
\ead[url]{https://FBreitinger.de}
\credit{Conceptualization, Writing - Original Draft, Writing - Review and Editing, Supervision}

\author[1]{Felix Freiling}[orcid=0000-0002-8279-8401]
\cormark[1]
\ead{felix.freiling@fau.de}
\credit{Conceptualization, Writing - Original Draft, Writing - Review and Editing, Supervision}

\address[1]{Department of Computer Science,
  Friedrich-Alexander-Universit\"at Erlangen-N\"urnberg (FAU),
  Erlangen, Germany}

\address[2]{School of Criminal Justice,
  University of Lausanne, 1015 Lausanne, Switzerland}

\cortext[1]{Corresponding authors.}

\title[mode=title]{As if Time Had Stopped -- Checking Memory Dumps for Quasi-Instantaneous Consistency}
\shorttitle{As if Time Had Stopped}
\begin{abstract}
	Memory dumps that are acquired while the system is running often contain
	inconsistencies like page smearing which hamper the analysis.  One
	possibility to avoid inconsistencies is to pause the system during the
	acquisition and take an instantaneous memory dump. While this is
	possible for virtual machines, most systems cannot be frozen and thus
	the ideal dump can only be quasi\hyph instantaneous, i.e., consistent despite
	the system running.
	In this article, we introduce a method allowing us
	to measure quasi\hyph instantaneous consistency and show both, theoretically,
	and practically, that our method is valid but that in reality, dumps can be but usually are not quasi\hyph instantaneously consistent. For the assessment, we run a pivot
	program enabling the evaluation of quasi\hyph instantaneous consistency for
	its heap and allowing us to pinpoint where exactly inconsistencies
	occurred.
\end{abstract}


\maketitle

\section{Introduction}


The acquisition and analysis of main memory are common tasks for
forensic investigators, e.g., to find encryption keys for storage or
analyze malware that only runs in memory. A common acquisition
procedure is to perform \emph{live memory acquisition}, i.e., to
utilize the software on the system under investigation to access and
dump memory. However, as the system is running and memory contents are
continuously updated by concurrent processes, the quality of such
snapshots is (at best) unclear.  A symptom of bad memory snapshots,
that is commonly observed, is \emph{page smearing} which is defined as
``an inconsistency that occurs in memory captures when the acquired
page tables reference physical pages whose contents changed during the
acquisition process'' \citep{pathforw}.  It is well-known that
established tools like Volatility have difficulties parsing low-quality
memory snapshots, resulting in cases where snapshots
cannot be analyzed at all. But what, actually, is a ``good'' memory
snapshot?


In practice, it is commonly accepted that freezing a system, i.e.,
stopping concurrent system activity before taking a memory snapshot,
produces the highest quality. Such snapshots are often referred to as
\emph{instantaneous} snapshots. Methods to create instantaneous
snapshots either have strong assumptions, e.g., assume that the
analyzed system runs as a virtual machine
\citep{martignoni2010,yu2012,kiperberg2019}, or are cumbersome to
execute, like cold boot attacks
\citep{DBLP:journals/cacm/HaldermanSHCPCFAF09,DBLP:journals/di/0004GF16}.
Therefore, in many practical situations memory acquisition is
necessarily performed live and the resulting snapshots are not
instantaneous. But in what sense can non-instantaneous snapshots be
compared regarding quality?

It has been observed \citep{introducing,defining} that certain
snapshots acquired live cannot be distinguished from instantaneous
snapshots. Such snapshots are called \emph{time-consistent}
\citep{introducing} or \emph{quasi-instantaneous}
\citep{defining}. By definition quasi-instantaneous snapshots avoid the many
hassles associated with live memory acquisition, but unless the memory
acquisition method itself provides consistency guarantees,
it was not known how memory snapshots can be tested for
quasi-instantaneous consistency. Clearly, such methods must rely on
some form of consistency indicators within the image. How these may
look like to precisely determine the consistency of a snapshot was so
far unclear. In this article, we describe a method to measure
quasi\hyph instantaneous consistency of memory snapshots based on
well-defined consistency indicators.  



\subsection{Related work}


After multiple works about the quality of
memory dumps \citep{Inoue:2011:VIT,Lempereur:2012:PAP,Campbell:2013:VMA},
three formal criteria for the assessment of a memory
dump's quality were defined by \citet{correctness}: \emph{correctness},
\emph{atomicity}, and \emph{integrity}.
Correctness is fulfilled if the contents of the memory dump are an exact copy of
the memory contents at the time of their acquisition. Atomicity addresses the
causal consistency of the memory dump. It depends on the cause-effect
relationships between memory accesses by different processes. The last criterion,
integrity, is assessed in relation to a point in time shortly before the memory
acquisition is started. Memory contents that change after this point in time
and before they were copied by the memory acquisition program lower the degree
of integrity of the memory dump. Two applications of the criteria for practical
evaluations of memory acquisition methods followed: one with a white-box testing
method \citep{evalplat}, and one with a black-box testing method
\citep{evaluatingat}.

In contrast to abstract measures such as atomicity, \citet{introducing} took a content-based approach to
assess the consistency of a memory dump. A memory dump is \emph{time-consistent}
if there ``exists a hypothetical atomic acquisition process that could have
returned the same result''. One method they applied in their evaluation to
assess the consistency of a memory dump is the number of virtual memory areas
(VMAs) that are attributed to a task by different sources. If the numbers differ
an inconsistency in a memory dump has been spotted.

Based on the idea of time consistency, \citet{defining} introduced two
formal criteria, \emph{instantaneous consistency}, and
\emph{quasi\hyph instantaneous} consistency. While the former
criterion portrays the ideal case for memory acquisition, pausing the
system's execution and copying all memory contents at the same time,
the latter can be fulfilled even if the system cannot be paused.  It
requires that the contents of the memory dump could have also been
acquired with a hypothetical instantaneous snapshot. Or in other
words, there was a time at which the dump's contents were
coexistent in memory. Therefore, a memory dump that fulfills the
latter criterion is as consistent as an instantaneous snapshot.  So
while quasi\hyph instantaneous consistency is as good as instantaneous
consistency, \citet{defining} fail to give a method to check or
observe it. Such a method would allow testing snapshots of benchmark
acquisition methods to gain trust in data and methods.

\subsection{Contributions}

In this paper, we devise a method with which (under certain
assumptions) it is possible to find out whether a portion of a snapshot is
quasi\hyph instantaneously consistent. Assumptions are the existence of
consistency indicators in memory. These represent information on the
last event that happened in a particular memory region and that potentially 
changed the content of that region. This extends the content-based approach of \citet{introducing}. 
Given such indicators, we show how it is
possible to test whether a snapshot is quasi\hyph instantaneously consistent.
Furthermore, if a memory dump is not quasi\hyph instantaneously consistent, we can
use the output of the method to assess the degree of inconsistency.

We present a formalization of the method and prove its correctness. We
also show how the necessary data structure for storing consistency
indicators can be implemented with increasingly efficient storage
requirements.
In a practical evaluation, we apply the method to frozen and live snapshots. As
expected, snapshots of frozen systems are always quasi\hyph instantaneously
consistent, those taken of live systems not necessarily.
In summary, the contributions of this paper are threefold: We provide
\begin{itemize}
\item a method to observe quasi\hyph instantaneous consistency,
\item a proof that it works theoretically, and
\item a proof-of-concept implementation that allows measuring
  consistency indicators in practice.
\end{itemize}

While we focus on main memory, our approach can naturally be applied
to situations in which other forms of storage (like persistent disk
storage) are acquired in a live fashion.

\subsection{Outline}

We first revisit the system model and previous consistency definitions
in Section~\ref{sec:definitions}. Our new method to observe and check
quasi\hyph instantaneous snapshots is presented in
Sections~\ref{sec:observing} and \ref{sec:checking}. Ways to improve
the memory efficiency of our method are discussed in
Section~\ref{sec:efficiency}. We provide the results of our practical
evaluation in Section~\ref{sec:evaluation} and discuss the results in
Section~\ref{sec:discussion}. We conclude in
Section~\ref{sec:conclusion}.

\section{Consistency of Snapshots}
\label{sec:definitions}

The consistency of a snapshot can be assessed from different
perspectives.  One is the causal perspective which takes into account
the active processes in the system and their causal relationships
\citep{correctness}. The basic idea of \emph{causal} consistency is
that the snapshot contains the cause for every effect. If the actions
of malware can be observed in the snapshot, all causally preceding
events must also be contained in the snapshot (e.g. the malware
infection). This definition is very generic and does not reference any
notion of real-time. As long as cause-effect relations are respected,
the system does not need to be frozen to acquire a snapshot that is
causally consistent.

The perspective we take in this article is more restrictive. We accept snapshots
 as consistent only if their contents were coexistent in memory at a previous
point in time. This consistency criterion is called \emph{quasi\hyph instantaneous
consistency} \citep{defining}.
To approach the formal definition of quasi\hyph instantaneous
consistency, we need to introduce some basic aspects of the system
model we assume.


\subsection{Model}\label{sec:model}
Based on \citet{correctness}, we define memory, events (modifying operations
on memory), and snapshots.

\paragraph{Memory}
We observe accesses to the set $R=\{r_1,\ldots,r_n\}$ of $n$ memory
regions. Intuitively, a memory region can be regarded as that part of
memory that can be acquired in one atomic action. Depending on the
real system, memory regions can consist of a single byte or a full
memory page.  Memory regions have values $v$ at specific points in
time. The sets of all possible values and points in time are denoted
$V$ and $T$, respectively. Memory can therefore be expressed by the
function $m: R \times T \rightarrow V$.

\paragraph{Events}
When a process performs an operation on a memory region this results
in an event $e$. We denote by $E$ the set of all such events. For any
event $e \in E$, $e.r$ denotes the memory region on which $e$
happened. An execution of the system is defined by a sequence of
events $\eta := [e_1,\ldots]$. As time between two events is of no
concern to our model, we define $T$ to be the set of natural numbers
$\mathbb{N}$. We assume that events generally \emph{change} memory
contents. So if no event happens on a region $r$ between times $t$ and
$t+n$, then the corresponding values in the memory are identical.
Formally:
$\forall r \in R, \forall t,n \in \mathbb{N}: m(r,t) = m(r, t+n)
\Leftrightarrow \forall k, t < k \leq t+n: e_k.r \neq r$.

\paragraph{Snapshot}

We formalize a snapshot as a function $s: R \rightarrow V \times T$,
i.e., for every memory region we store the value and the time at which it was
copied. We denote by $s(r).v$ the value stored
for region $r$ in snapshot $s$ and by $s(r).t$ the corresponding
time. Note, as established above, the time $t$ advances whenever an event is
executed. The vector containing all values in all regions of the snapshot is denoted $V_s :=
[s(r_1).v, \ldots, s(r_n).v]$, the vector containing all times $T_s
:= [s(r_1).t, \ldots, s(r_n).t]$.

The model can be visualized using space/time diagrams \citep{virtual}. An
example for a system with three memory regions, $r_1$, $r_2$, and $r_3$ is shown
in Fig.~\ref{fig:model}. The arrows represent the memory regions over time,
events, $e_1$, $e_2$, $e_3$, and $e_4$ in the example, are denoted by black
dots.  The time at which a memory region is copied in a snapshot is denoted with
a rectangle. The rectangles belonging to one snapshot are connected to each
other. In the example two snapshots, $s_1$ and $s_2$, can be seen.

\begin{figure}
  \centering
  \includegraphics[scale=0.6]{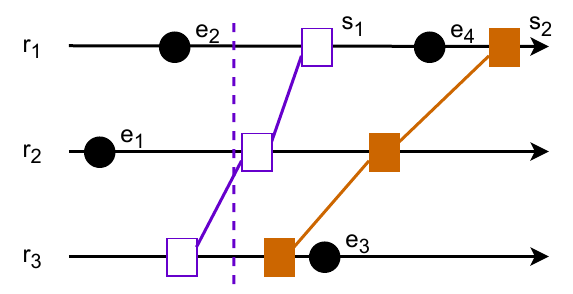}
  \caption{With space/time diagrams the events that take place on
    memory regions over time can be visualized. Each region is
    represented by an individual timeline (time passes from left to
    right), events are denoted as black dots and the acquisition of a
    memory region in a snapshot as rectangle.}
  \label{fig:model}
\end{figure}

\subsection{Quasi-instantaneous consistency}

\citet{defining} defined the following notions of consistency.
The ideal case for a snapshot is that it is taken \emph{instantaneously}. In a
snapshot that satisfies instantaneous consistency every memory region was copied
at the same time. 

\begin{definition}[instantaneous consistency]
  A snapshot $s$ satisfies \emph{instantaneous consistency} iff all
  memory regions in $s$ were acquired at the same point in
  time. Formally:
  $$ \forall r, r'\in R: s(r).t = s(r').t$$
  If s satisfies instantaneous consistency we call s instantaneous.
\end{definition}

When a system cannot be frozen it might still be possible to acquire a snapshot
with the same contents as if it had been taken instantaneously. In this case the 
content is identical to an instantaneous snapshot (although not taken instantaneously) 
and we call such a snapshot quasi\hyph instantane\-ously consistent.

\begin{definition}[quasi-instantaneous consistency]
	\label{def:quasi}
  A snapshot $s$ satisfies \emph{quasi\hyph instantaneous consistency} iff
  the values in the snapshot could have also been acquired with an
  instantaneous snapshot $s'$. Formally:
  %
  $$
	  \exists s': s' \textrm{is instantaneous} \mathrel{\land} (V_{s'} = V_s)
  $$
  %
  If s satisfies quasi\hyph instantaneous consistency we call s quasi\hyph instantaneous.
\end{definition}

Two example snapshots are shown in Fig.~\ref{fig:model}. Snapshot
$s_1$ is quasi\hyph instantaneous since an instantaneous snapshot can
be found that would have had the same contents. Such an instantaneous
snapshot could have been taken right after event $e_2$ took place
and is indicated by a dashed vertical line. For the second snapshot,
$s_2$, on the other hand, it is not possible to construct an
instantaneous snapshot with the same contents.  The reason is that
event $e_3$ happened before $e_4$ but in the snapshot the changes made
by $e_3$ cannot be seen while those made by $e_4$ are included.  Thus,
snapshot $s_2$ is \emph{not} quasi\hyph instantaneous.

\section{Observing Quasi-Instantaneous Consistency}
\label{sec:observing}

Our approach to observe quasi\hyph instantaneous consistency is based
on the observation of \emph{consistency indicators} within the
snapshot. Oftentimes, such indicators already exist as part of the
running system. For example, kernel data structures that save
redundant information can serve as indicators \citep{introducing}.
However, artificial consistency indicators can also be deployed as
part of general forensic readiness procedures or within the memory
manage\-ment of individual processes.

\subsection{Current time and time of last event}\label{sec:idea}

\begin{figure}
  \centering
  \includegraphics[scale=0.6]{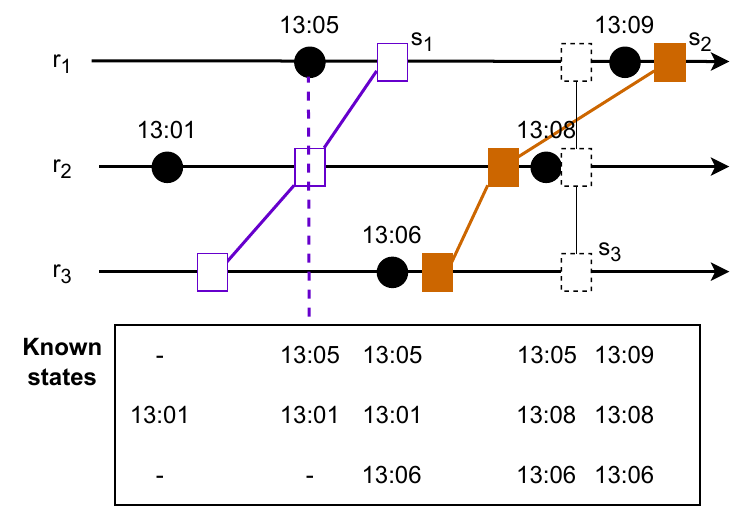}
  \caption{Knowledge about the times at which events occurred on memory regions
	allows to determine if a snapshot is quasi\hyph instantaneous. Given a
	snapshot, a vector of the times at which the last event relative to the
	snapshot occurred can be formed. If this vector matches one of the known
	sates, the snapshot is quasi\hyph instantaneous.}
  \label{fig:idea}
\end{figure}

If we want to know exactly in which memory regions inconsistent
contents are located, knowledge about previous states of the memory
contents is necessary. An example is shown in Fig.~\ref{fig:idea} where 
events happen in real-time and the timestamps of events are recorded in a 
table shown below the space/time diagram. Note, such a data structure of all event
timestamps enumerates all possible instantaneous snapshots since
memory contents only change through events. For example, the
instantaneous snapshot taken right after the event at 13:08, $s_3$ in the
figure, would contain data as changed by the events at
13:05 (on region $r_1$), 13:08 (on region $r_2$) and 13:06 (on region $r_3$).

To identify if a snapshot is quasi\hyph instantaneous, we need to find
the ``matching'' instantaneous snapshot in the list of all
instantaneous snapshots described above. To do this, we can either
``scan'' the data structure from beginning to end, or search in the
vicinity of the timestamps that are stored in the snapshot. For a more
specific search, the ability to determine the time of the last event
on each memory region relative to the time at which the snapshot was
taken on that region is helpful. For example, for snapshot $s_1$, the
time of the most recent event on $r_1$ is 13:05. If the vector of
these time stamps matches one of the possible instantaneous snapshots
listed in the data structure, the snapshot is quasi\hyph
instantaneously consistent.

To illustrate the idea, Fig.~\ref{fig:idea} depicts two snapshots
$s_1$ and $s_2$; $s_1$ is quasi\hyph instantaneous since the vector of
the last events matches the known state added at time 13:05. For $s_2$
the searched vector is \{13:09, 13:01, 13:06\}.  Since this state is
not contained in the known state array, the snapshot is not quasi\hyph
instantaneous.




\subsection{Two-dimensional global counter array}\label{sec:twodim}

We formalize this idea based on state information stored in unique
counters saved at each memory region access in a global structure, the
\emph{global counter array}, and the region itself. The global counter array can
be implemented in different variations. We introduce the general idea first,
followed by a variant of the global counter array that carries redundant
information helpful for visualization.

\paragraph{Global counter array}
The global counter array $G$ is a two-dimensional array $R \times T$. As
defined in section~\ref{sec:model} $T = \mathbb{N}$. It
contains a row for each $r \in R$. Its rows and columns are initialized with
zero. Since $T$ is infinite, theoretically, $G$ is also an infinite data
structure. However, at any finite point in time $G$ is also finite.

The current column to write to in $G$ is identified using the current logical
time $t \in \mathbb{N}$. It is initialized with zero. Algorithm~\ref{alg:update}
shows the sequence of actions triggered by an event $e$ on $r_i$.
When a memory region $r_i$ is accessed the time $t$ is incremented by one and a
value $x$ written to $G$ at the index $t$: $G[r_i][t] := x$. The value $x$ is
dependent on the implementation variant chosen for the global counter array as
we will see later. The time $t$ is saved in the memory region $r_i$ on which the
event occurred.

\algblockdefx[Event]{Event}{EndEvent}
{\textbf{Upon} }
{\textbf{Update finished}}
\algtext*{EndEvent}
\begin{algorithm}
\caption{Sequence of actions triggered by an event $e$ on memory region $r_i$}
	\label{alg:update}
\begin{algorithmic}
	\Event Event $e$ on $r_i$
	\State $t := t+1$
	\State $G[r_i][t] := x$
	\State Save $t$ in $r_i$
	\EndEvent
	\end{algorithmic}
\end{algorithm}

\paragraph{Current time}
We denote the vector $G_t$ which contains the index of the last visible status
update for each $r$ in the global counter array $G$ at a logical point in time
$t$ the \emph{current time} of $t$. The value in the vector at index $i$ is
returned by $G_t[i]$.


\begin{figure}
  \centering
  \includegraphics[scale=0.6]{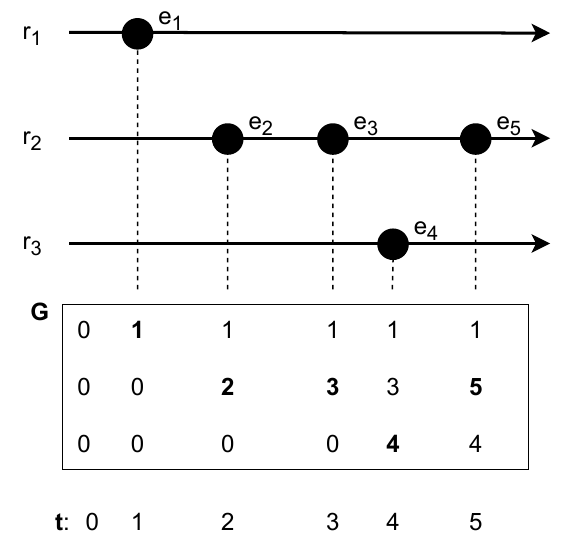}
  \caption{Each time an event happens on a memory region $r_i$, $t$ is increased by
	one and its value written to the appropriate row $r_i$ in the global counter
	array $G$.}
  \label{fig:carry}
\end{figure}

\subsection{Carry along global counter array}

One possibility for keeping track of coexistent states is to save the value of
$t$ for each event on a region $r$ in $G$ and carrying along the last visible counter
updates for all other regions. Obviously, this representation is also a direct representation of all possible instantaneous snapshots with logical time.

When an event $e$ occurs on memory region $r_i$, the sequence of actions shown
in Algorithm~\ref{alg:update} is followed: First, the time $t$ is incremented by
one. The value $x$ is written to $G$ at the index $t$, i.e., $t$: $G[r_i][t] := t$.
Then the time $t$ is saved in $r_i$ as well. Additionally, for all other $r$, the
value at index $t-1$ is written to $G$, thereby carrying along the values of
previous updates: $G[r_i][t] := G[r_i][t-1], \forall r \in R: r \neq e.r$.
An example of how $G$ is updated for each event is shown in
Fig.~\ref{fig:carry}.

\paragraph{Current time}
The current time for a logical point in time $t$ is reconstructed from
the values for each row in $G$ at index $t$:
$G_t := {G[r_1][t],\ldots, G[r_n][t] }$.  For example, the current
time in Fig.~\ref{fig:carry} at time $t=5$ is $G_t=(1,5,4)$. This time
can be used to check for quasi-instantaneous consistency, as we now
explain.

\section{Checking Quasi-Instantaneous Consistency}
\label{sec:checking}

The question is how to verify if a snapshot is quasi\hyph instantaneously 
consistent. For this purpose, it needs to be determined if a
point in time exists at which the same contents were coexistent in memory as in
the snapshot. Comparing the last point in time saved in each memory region to
the states saved in the global counter array allows us to do this.

Given that the value of each memory region $r$ in the snapshot $s$ is defined by the last
event that occurred on the region, each $s$ is equivalent in its values
to the normalized snapshot $N(s)$, which is taken right after the occurrence of the last
event for each memory region $r$ from the point of view of $s$.
Fig.~\ref{fig:norm} shows an example of a snapshot and its normalized form.

\begin{figure}[t]
  \centering
  \includegraphics[scale=0.6]{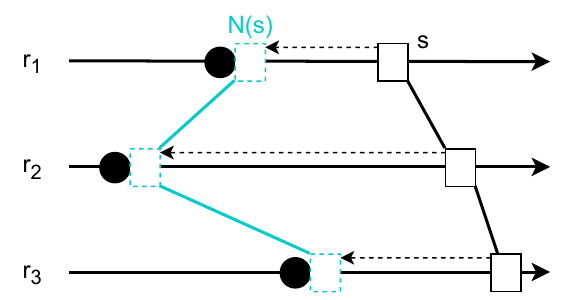}
	\caption{Given a snapshot $s$, its normalized form $N(s)$ contains the
	same time stamps as a snapshot which copied each memory region at
	the moment the last event from the perspective of $s$ had happened on
	it.}
  \label{fig:norm}
\end{figure}

\begin{definition}[Normalized snapshot $N(s)$]
For $r \in R$, we define $t'_r$ as the point in time at which the last event
relative to a snapshot $s$ was executed on $r$: $t'_r := \max(\{0\} \cup \{ i \leq s(r).t
\;|\; e_i.r = r\})$.
For each memory region $r$ the snapshot $N(s)$ contains the appropriate point in
time $t'_r$ and the value saved in the memory region at point in time $t'_r$:
$N(s)(r) := (t'_r, m(r, t'_r))$.
\end{definition}

\begin{proposition}\label{prop:normal}
	The values of $N(s)$ and $s$ are equivalent: $V_{N(s)} = V_s$
\end{proposition}
\begin{proof}

	We want to show that $\forall r \in R: N(s)(r).v = s(r).v$.

	Fix an $r \in R$ and let $t := s(r).t$. The case $t'_r = t$ is trivial,
	since then $m(r, t'_r) = m(r, t)$ by
	assumption. As such, $V_{N(s)} = V_s$.

	Consider $t'_r < t$. If we assume $m(r, t'_r) \neq m(r, t)$,
	an event $e_z$ occurs after $t'_r$ at time $z$: $e_z.r =
	r$ and $t'_r < z \leq t$. It follows that $z \in \{ i \leq s(r).t \;|\; e_i.r =
	r\}$, but then $t'_r$ is not the time at which the last event
	happened on region $r$, since $z > t'_r$. This contradicts the
	definition of $N(s)$. Hence, $m(r, t'_r) = m(r, t)$ and $V_{N(s)} = V_s$
	accordingly.

\end{proof}

Thus, when looking at a snapshot $s$, it is equivalent to look at
$N(s)$ instead. When comparing two snapshots, $s_1$ and $s_2$, they can be
substituted with $N(s_1)$ and $N(s_2)$, respectively. This makes comparisons
easier, since the values stored in the normalized snapshots are equal iff the
times are equal.

\begin{proposition}\label{prop:nvaluetime}
For two snapshots $s_1$ and $s_2$: $T_{N(s_1)} = T_{N(s_2)} \Leftrightarrow
	V_{N(s_1)} = V_{N(s_2)}$
\end{proposition}

\begin{proof}

\noindent ($\Rightarrow$) Let $T_{N(s_1)}$ be equal to $T_{N(s_2)}$:
		If for both $N(s_1)$ and $N(s_2)$ the time at which the last event
		which occurred on a region $r \in R$ is equal, $t :=
		N(s_1)(r).t = N(s_2)(r).t$, then we know that $m(r,
		N(s_1)(r).t) = m(r, t) = m(r, N(s_2)(r).t)$. As such,
		$V_{N(s_1)}
			= V_{N(s_2)}$.

	\noindent ($\Leftarrow$) Let $V_{N(s_1)}$ be equal to $V_{N(s_2)}$: According to Proposition
		\ref{prop:normal}, given an $r \in R$, we have
		$N(s_1)(r).v = s_1(r).v = s_2(r).v = N(s_2)(r).v$.	
		We need to show $N(s_1)(r).t = N(s_2)(r).t$.
		
		If the values of $s_1$ and $s_2$ are equal, $s_1(r).v =
		s_2(r).v$, this means $m(r,s_1(r).t) =
		m(r,s_2(r).t)$.
		Let $t_1 := s_1(r).t$, $t_2 := s_2(r).t$.
		W.l.o.g. $t_1 < t_2$.
		Then an $n \in \mathbb{N}$ exists for which $t_2 = t_1 + n$.
		Since $m(r,t_1) = m(r,t_2)$, there has been no $e_k$ where $t_1
		< k \leq t_1 + n$ with $e_k.r = r$.

		Then the sets $\{i \leq t_1\ |\ e_i.r = r\}$ and $\{i \leq t_1 +
		n\ |\ e_i.r = r\}$ are equal, the latter of course being $\{i
		\leq t_2\ |\ e_i.r = r\}$. Now $N(s_1)(r).t = max \{\{0\} \cup
		\{i \leq t_1\ |\ e_i.r = r\}\} = max \{\{0\} \cup \{i \leq t_2\
		|\ e_i.r = r\}\} = N(s_2)(r).t$, which completes the proof.

\end{proof}

Now that we know that for two normalized snapshots, their values are only equal iff
their times are equal, the question remains how we can use this to check
snapshot $s$ for quasi\hyph instantaneous consistency. The missing piece to perform the
check is the associated instantaneous snapshot of $s$, denoted $\hat{s}$.
Fig.~\ref{fig:asso} shows a snapshot $s$ and its associated instantaneous snapshot $\hat{s}$
taken at the highest time found in the normalized snapshot $N(s)$.

\begin{definition}[Associated instantaneous snapshot of $s$]
For a snapshot $s$, we denote by $\hat{s}$ the instantaneous snapshot at
$\hat{t}_s$ which we call the \emph{associated instantaneous snapshot}
of $s$, where $\hat{t}_s := \mathsf{max}(\{N(s)(r).t \;| \; r \in R\})$.
Note that $T_{N(\hat{s})} = G_{\hat{t}_s}$.
\end{definition}

\begin{figure}[t]
  \centering
  \includegraphics[scale=0.6]{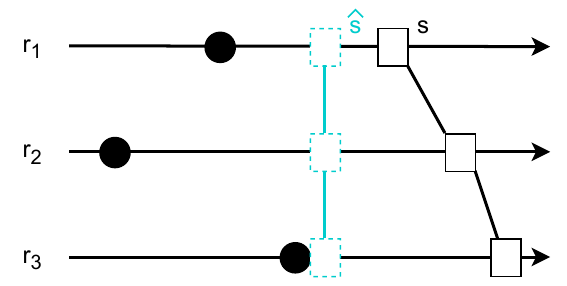}
	\caption{The associated instantaneous snapshot $\hat{s}$ of a snapshot $s$ is the
	instantaneous snapshot taken at the highest time of its normalized
	snapshot $N(s)$.}
  \label{fig:asso}
\end{figure}

If a snapshot is equal in its values to its associated instantaneous snapshot it is
quasi\hyph instantaneously consistent. As we have established that two normalized
snapshots will be equal in their values iff their times are equal, we can perform
the comparison based solely on the times of the normalized snapshot of $s$,
$N(s)$, and the normalized snapshot of its associated instantaneous snapshot, $N(\hat{s})$.

\begin{theorem}
	A snapshot $s$ is quasi\hyph instantaneously consistent iff $T_{N(s)} =
	T_{N(\hat{s})}$.
\end{theorem}

\begin{proof}
\noindent ($\Rightarrow$) Given an instantaneous snapshot $s'$ for
			which $V_{s'} = V_s$ , we show
			that $T_{N(s)} = T_{N(\hat{s})}$.

			Since $V_{s'}= V_s$, according to Proposition
			\ref{prop:normal} we can use the normalized snapshots
			instead: $V_{N(s')} = V_{N(s)}$. It follows that
			$T_{N(s')} = T_{N(s)}$ according to Proposition
			\ref{prop:nvaluetime}. Therefore $\hat{t}_{s'} =
			\hat{t}_s$. Since there is only one instantaneous
			snapshot at any given time $t$, $\hat{s'} = s' =
			\hat{s}$. Thus, when we substitute $s'$ by $\hat{s}$,
			$T_{N(s)} = T_{N(\hat{s})}$. 

\noindent ($\Leftarrow$) Given $T_{N(s)} = T_{N(\hat{s})}$ we show
			that a snapshot $s'$ exists which is instantaneous and
			for which $V_{s'} = V_s$.

			Let $s' := \hat{s}$. By definition $\hat{s}$ is
			instantaneous. According to Proposition \ref{prop:nvaluetime},
			$T_{N(s)} = T_{N(\hat{s})} \Rightarrow V_{N(s)} =
			V_{N(\hat{s})}$. Therefore, according to Proposition \ref{prop:normal},
			$V_s = V_{\hat{s}}$.

\end{proof}

With this theorem, we have shown that we can use the normalized snapshot instead
of the original snapshot to evaluate if the snapshot is quasi\hyph instantaneously consistent
or not. It also becomes apparent that comparing the time is sufficient to
establish if the values of the snapshot were coexistent in memory at some point
in time. Since $T_{N(\hat{s})} = G_{\hat{t}_s}$ it also shows how the states saved
in the global counter array are used to determine existent states. Algorithm
\ref{alg:quasi} summarizes how the check is performed.

\begin{algorithm}
\caption{Checking for quasi\hyph instantaneous consistency}\label{alg:quasi}
\begin{algorithmic}
	\State Compute $N(s)$ \Comment Extract time $t$ saved in each region
	\State $\hat{t}_s := \mathsf{max}(\{N(s)(r).t \;| \; r \in R\})$
	\State Compute $G_{\hat{t}_s}$	\Comment See Algorithm \ref{alg:spec}
	\State $T_{N_{\hat{s}}} := G_{\hat{t}}$
	\If{$T_{N_{\hat{s}}} = T_{N(s)}$}
		\State $s$ is quasi\hyph instantaneously consistent
	\Else
		\State $s$ is not quasi\hyph instantaneously consistent
	\EndIf
	\end{algorithmic}
\end{algorithm}

\section{Improving Memory Efficiency}
\label{sec:efficiency}


The implementation of the global counter array as shown in
Section~\ref{sec:twodim} is inefficient both computationally and
regarding memory usage. In the following, we first show a more efficient two-dimensional
implementation. As it becomes apparent that one dimension is enough to carry the
necessary information, we then present a one-dimensional variant we used for
implementing the global counter array.

\subsection{Simplified global counter array}

Looking at Fig.~\ref{fig:carry} it becomes apparent that a lot of redundant
information is saved in the global counter array $G$ because the index at which $t$ is written and its
value are identical. Additionally, from a practical perspective, it is more
efficient to not carry along previous values of $t$. Instead, when an update of
$t$ occurs a $1$ is written at index $t$ for the appropriate $r$. For all other
$r$ the initial value, $0$, is not changed.

For the simplified version, an event on memory region $r_i$ triggers the sequence of actions shown
in Algorithm~\ref{alg:update}: First, the time $t$ is incremented by one.
Then, a value $x$ is written to $G$ at the index $t$, for the simplified global
counter array $x$ is always $1$: $G[r_i][t] := 1$. Lastly,
the value of $t$ is saved in $r_i$. Here, no additional steps are necessary.
Fig.~\ref{fig:simpler} shows the same sequence of events as in
Fig.~\ref{fig:carry} but with the adapted implementation of $G$. 

\begin{figure}
  \centering
  \includegraphics[scale=0.6]{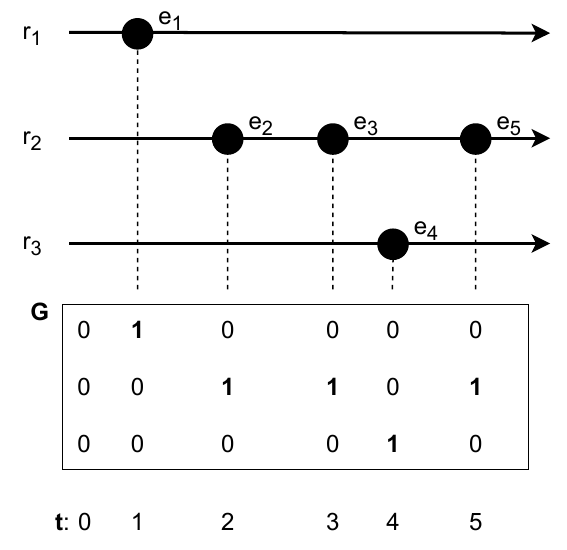}
  \caption{Each time an event happens on a memory region $r_i$, $t$ is increased by
	one and $1$ written to the according row in the global counter
	array $G$.}
  \label{fig:simpler}
\end{figure}

\paragraph{Current time}
Because the last updates are not carried along, reconstructing $G_t$ requires to
find the last update for all memory regions $r$ except the one at which a $1$
can be found in $G$ for time $t$. This can be done as shown in Algorithm~\ref{alg:spec}.

\begin{algorithm}
	\caption{Computing the current time $G_t$ for the logical time
	$t$}\label{alg:spec}
\begin{algorithmic}
	\State Initialize vector $G_t$ with $-1$
	\State $t_i := t$
	\While{$\exists i: G_t[i] = -1 $}
		\State Find row $r_h$ where ($G[r_h][t_i] = 1) \wedge
		(G_t[r_h] = -1)$
		\State $G_t[r_h] := t_i$
		\State $t_i := t_i - 1$
		\If{$t_i = 0$}
			\ForAll{ $G_t[i]$ for which $G_t[i] = -1$}
				\State $G_t[i] := 0$
			\EndFor
		\EndIf
	\EndWhile
	\end{algorithmic}
\end{algorithm}

\subsection{One-dimensional global counter array}

The implementation variant of $G$ shown in
Fig.~\ref{fig:simpler} uses more memory than necessary to carry
the needed information. Although at each logical point in time $t$ only one
memory region is updated, for all other regions memory is reserved with only
zeros entered. Since we only need to save information for exactly one region per
logical point in time, we can save the known states in a list instead of a
two-dimensional array. Since no second dimension exists to indicate the region
at which the event occurred, this information needs to be saved in the list.

The index to write to in the one-dimensional global counter array $G$ remains
the time $t$. When a memory region $r_i$ is accessed, $t$ is incremented and $i$
saved in $G$ at index $t$: $G[t] := i$. The value of $t$ is saved in $r_i$.
An example of the adapted global counter array with the same sequence of events as in the
previous two examples is shown in Fig.~\ref{fig:list}.

\begin{figure}
  \centering
  \includegraphics[scale=0.6]{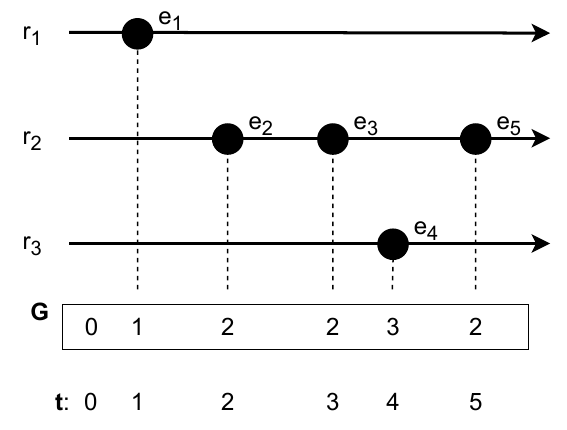}
  \caption{Each time an event happens on a memory region $r_i$, $t$ is increased
	by one and the number of the region, $i$, written to the global counter
	array $G$.}
  \label{fig:list}
\end{figure}

\paragraph{Reconstructing a specific current time}
To reconstruct the coexistent values at a logical point in time $t$,
entries for each memory region in $G$ at or before $t$ need to be searched. If no
entry for a memory region can be found no events occurred yet which means the
saved value in the region equals zero. Algorithm~\ref{alg:list} shows the
detailed procedure.

\begin{algorithm}
	\caption{Computing the current time $G_t$ for the logical time
	$t$ based on the one-dimensional global counter array}\label{alg:list}
\begin{algorithmic}
	\State Initialize vector $G_t$ with $-1$
	\State $t_i := t$
	\While{$\exists i: G_t[i] = -1 $}
		\State $r := G[t_i]$
		\If{$G_t[r] = -1$}
			\State $G_t[r] := t_i$
		\EndIf
		\State $t_i := t_i - 1$
		\If{$t_i = 0$}
			\ForAll{ $G_t[i]$ for which $G_t[i] = -1$}
				\State $G_t[i] := 0$
			\EndFor
		\EndIf
	\EndWhile
	\end{algorithmic}
\end{algorithm}

\section{Evaluation}
\label{sec:evaluation}

Now that we have shown that, given proper consistency indicators, theoretically quasi\hyph instantaneous
consistency can be observed, we present a practical
proof-of-concept application of the method. It allows observing the
quasi\hyph instantaneous consistency of memory regions in one process. We consider
two main system states for the evaluation, frozen and running. We expect that
memory dumps taken of frozen systems satisfy quasi\hyph instantaneous consistency,
while those taken concurrently to the running system are expected to not 
(necessarily) be consistent. The evaluation is performed with a semi-automated procedure,
described subsequently, which mainly differs in the method chosen to create the
memory dump depending on the system state.

\subsection{Procedure}

For the evaluation of memory dumps taken in both system states, we use a virtual
machine running Ubuntu 18.04 with 4\,GB of RAM and 4 CPUs. Quasi\hyph instantaneous
consistency is observed in a specifically crafted \emph{pivot program}.
It meets the requirements to apply the method practically: 
\begin{enumerate}
	\item The ability
to observe accesses to memory regions 
	\item The ability to write counter values
to memory regions
	\item Enough memory for the global counter array.
\end{enumerate}
In the pivot program memory regions are represented by list elements and changes on them are
tracked in a one\hyph dimensional global counter array. The array is allocated
with a fixed size that is sufficient to capture the events taking place during
the intended runtime of the program. The changes on the list elements
are performed by one or more threads. The threads randomly choose a list element to
remove from the synchronized list and after a short wait reinsert the element at
the beginning of the list. Each update (insertion/removal) of a list element causes an
update of the time of the last event in the list element and the global
counter array. The number of list elements and threads is set at the program start.

Memory dumps of the live system are taken for two different levels of activity,
\emph{low} and \emph{high}. When creating memory dumps for the low activity
level only the pivot program is executed. In comparison, the high activity level 
executes several other programs in parallel. This level of
activity is also generated for the frozen system snapshots.
Memory dumps taken in the frozen and the live system state with high activity
can be summarized as follows (manually performed actions are labeled
with numbers, automated actions with letters):
 
\begin{enumerate}
	\item Start VM
		\begin{enumerate}
			\item Start pivot program
			\item Mount shared folder
			\item Start grep: \texttt{timeout 2m grep -r "libc" /  \&}
			\item Retrieve meta info of pivot program (pid, heap
				range)
			\item Move meta info to shared folder
		\end{enumerate}
	\item Open Firefox
	\item Open YouTube, click on video
	\item Open LibreOffice Writer, continuously write text
		\begin{enumerate}[resume]
			\item Take memory dump
			\item Dump pivot program's heap contents
		\end{enumerate}
\end{enumerate}

The memory dump is taken approximately one minute after the \texttt{grep}
command was executed. For all memory dumps taken without freezing the system,
as a last step the heap contents of the pivot program are dumped using
\texttt{gdb}.

Using this procedure, 30 memory dumps were created in total. Ten for the frozen
system state with high activity, ten for the live system state with high
activity, and ten for the live system state with low activity (details of how
these memory dumps were acquired are given below). All memory dumps and
analysis results, as well as the scripts used for their analyis and the source
code of the pivot program are publicly
available\footnote{\url{https://zenodo.org/record/8089517}}. The number of active
threads in the pivot program is set to eight for high activity. The memory dumps
of a live system with low activity are taken with the same timing as those with
high activity but without steps (c), (2), (3), and (4), and the number of active
threads in the pivot program is set to one instead of eight.

\subsection{Analysis}
In the analysis two types of inconsistencies are evaluated, quasi\hyph instantaneous
inconsistencies in the pivot program's heap, and inconsistencies between numbers
of virtual memory areas (VMAs) saved for each process by the kernel.

\paragraph{Quasi-instantaneous inconsistencies} When searching for quasi\hyph instantaneous
inconsistencies in the pivot program's process address space, first its heap is extracted from the
memory dump using the Volatility plugin \texttt{linux\_dump\_map}. Next, the
list elements and the time of the last event on them as well as the global
counter array are retrieved from the heap pages using a python script. In the
case of the memory dumps taken without freezing the system, the global counter
array is instead retrieved from the heap dump taken with \texttt{gdb}. This is
necessary as, while in the virtual process memory the global counter array is
located after the list elements and their counters, in the physical memory they
might be jumbled. If the global counter array is acquired before the list
elements, its contents could be not up to date with the last changes made on the
list elements.

To check for violations of quasi\hyph instantaneous consistency, we follow the
steps of Algorithm~\ref{alg:quasi}: The time of the last
event for each region is saved in a vector which is equivalent to the normalized
snapshot $N(s)$. Then, the maximal time stamp in this vector $\hat{t}$ is
identified, upon which the current time $G_{\hat{t}}$ is computed from the global counter array.
Lastly, the normalized snapshot and the current time, i.e., the snapshot's
associated instantaneous snapshot, are compared. Should they differ in one or
more values, a violation of quasi\hyph instantaneous consistency has been
identified.

\paragraph{VMA inconsistencies}
To gain insight into inconsistencies in kernel data structures, we use a method
suggested by \citet{introducing}. The number of VMAs assigned to each process
can be retrieved from different sources, a linked list of VMAs managed for each
process, a red-black tree of the VMAs, and the counter of assigned VMAs saved
for each process in its \texttt{task\_struct} structure. The Volatility plugin
\texttt{linux\_validate\_vmas}\footnote{Published by the authors at
\url{https://github.com/pagabuc/atomicity_tops}.} retrieves the number of VMAs
from the three sources and compares them. If a mismatch is detected, the
name of the corresponding process and the different values are returned. The
total number of processes with inconsistent VMA numbers is gathered for each
dump.

\subsection{Frozen system}

To take a memory dump of the frozen system, we use \texttt{virsh dump} with
option \texttt{-{}-memory-only}. As this command has to be performed by the
host, a script is started on the host once the shared folder has been
mounted that executes the command after one minute.
In the ten created memory dumps, as expected, no quasi\hyph instantaneous
or VMA inconsistencies were found. All ten snapshots were quasi\hyph instantaneously consistent.

\subsection{Running system}

\begin{table*}[pos=ht]
	\centering
	\begin{tabular}{lllcccc}
                \toprule
		System state & Inconsistency type & Activity & Min & Max & Average & Affected
		dumps\\
                \midrule
		\multirow{2}{*}{Frozen} & Quasi\hyph instantaneous &
		\multirow{2}{*}{High}  & 0 &0 &0 & 0/10 \\
		& VMA & & 0 & 0 & 0 & 0/10 \\ \midrule
		\multirow{4}{*}{Live} & \multirow{2}{*}{Quasi-instantaneous} & Low & 0 & 3 & 0.8 & 5/10 \\ 
		& & High & 0 & 37 & 13.8 & 7/10 \\ 
		&\multirow{2}{*}{VMA} & Low & 0 & 1 & 0.1 & 1/10\\
		&& High & 3 & 7 & 4.9 & 9/9 \\
                \bottomrule
        \end{tabular}
        \caption{The table shows the minimum, maximum and average number of
	quasi\hyph instantaneous and VMA inconsistencies found in the 30 memory dumps
	created for the evaluation. Out of the memory dumps taken with
	high system load one could not be analyzed regarding VMA
	inconsistencies. Therefore the average number of inconsistencies is
	calculated for nine instead of ten dumps.}
        \label{tab:incons}
\end{table*}

We use LiME with option \texttt{format=lime} to take memory dumps of
running systems. This is done from within the VM and integrated into the same
script that performs the other automated tasks.
For each activity (low and high), ten memory dumps were taken. The
observed inconsistencies are summarized in Table~\ref{tab:incons}.

For low system activity, fewer inconsistencies occurred than for
high system activity. The number of memory dumps affected by quasi\hyph instantaneous
inconsistencies is higher than the number of dumps in which VMA inconsistencies
were found.

With higher activity, the number of inconsistencies rises distinctly. Seven out
of ten memory dumps contain quasi\hyph instantaneous inconsistencies. Out of
them one only contained two inconsistencies, the others 15 or more.
The three memory dumps that contain no quasi\hyph instantaneous
inconsistencies, and the one with only two are noteworthy compared to the
average number of found inconsistencies. For VMA inconsistencies a similar
observation can be made. The nine memory dumps that could be analyzed
regarding VMA inconsistencies, contained them. They are attributed to processes
related to the web browser, audio, and gnome-shell.
One memory dump had to be excluded from the VMA inconsistency check as 
during the check for VMA inconsistencies in this dump, the function
used in the \texttt{linux\_validate\_vmas} Volatility plugin to traverse the red
black tree did not return and the plugin's
execution had to be stopped. Memory dumps for which the plugin did
not terminate were also observed by \citet{introducing}. While this is
probably one symptom of inconsistencies in the memory dump, no statement about
the number of VMA inconsistencies for this memory dump can be made.

\section{Discussion}
\label{sec:discussion}

As expected most of the memory dumps taken on the live system with high activity
are not quasi\hyph instantaneously consistent. But the numbers vary
noticeably and three memory dumps did not have inconsistencies. Using the
information available through the check for quasi\hyph instantaneous consistency, we
can perform some further examinations. Given the small number of memory dumps,
we do not claim that these results can be generalized to any memory acquisition
but they show the advantages of using our method when investigating the reasons
for inconsistencies in a memory dump.

\paragraph{Reconstruction of physical addresses:} While checking the heap of the pivot
program for inconsistencies, the virtual addresses of the list elements and the
global counter array were gathered. They are ordered sequentially on adjacent
pages in the virtual memory but their mappings to physical pages do not have to
be in the same order or address range.
Therefore, we reconstructed to which physical pages they were mapped using
Volatility's \texttt{linux\_memmap} plugin. From these mappings, we could
reconstruct the range of physical addresses in which the list elements and the
global counter array were located. This allows us to calculate the distance of each address
to the nearest next address (i.e, the nearest list element).
Table~\ref{tab:dist} summarizes the findings per memory dump ordered by the
number of found quasi\hyph instantaneous inconsistencies in them.
\emph{Range (in pages)} is the size of the physical address range in which the
list elements and global counter array are located, displayed as number of pages.
It is calculated by subtracting the lowest found address from the highest. The
\emph{distances} columns include the number of list elements that were within a 10
pages radius or directly neighbors, respectively. The largest found distance
between two elements is given by \emph{Max distance}.
All distances are given as the number of pages (the page size is 4096 bytes).

\paragraph{Spread is bad:}
The table supports the intuition that a longer
range in which the addresses are distributed will likely lead to more
inconsistencies. 
Or vice-versa, in memory dumps with fewer inconsistencies, more
contents of interest are located on adjacent pages than in those with more
inconsistencies. 

\paragraph{Details - dump \#1:} 
It has the highest number of inconsistencies but a relatively high number of
adjacent physical addresses and not the largest range of addresses. The maximal
distance between two list elements is however the largest one in the evaluated
memory dumps. Taking a look at the location of the list element for which
the current time was calculated reveals that it is located towards the end of
the memory range and is separated by the observed largest distance from the
previous 93 list elements. Thus, it is likely that there was a longer time frame
during the acquisition between copying the previous elements and the last ones.
Combining this with the earlier acquisition of most of the list elements, it
becomes likely that updates on them are missed. 

\paragraph{Details - dump \#6:}
Here, the range is the third
smallest, and 85 of 101 addresses have a distance between one and ten pages but
still 15 inconsistencies occurred. A closer look at the list of elements
for which inconsistencies occurred in this memory dump provides a possible
explanation. They are located more in the beginning of the address range with
mostly smaller distances between them while the list element with the highest
time stamp, i.e. the one for which the current time was identified, is located
more towards the end. A big gap, 72\,745 pages (the largest distance plus some
smaller gaps afterward), lies between the last list
element with an inconsistency and the one with the highest time stamp.
Therefore, similarly to dump \#1, it becomes more likely that changes on the earlier list elements
occur before this list element is acquired. From the difference between the
counters of the list elements for which inconsistencies were detected and their
values in the current time we can also see that many updates occurred on them,
the smallest number of missed updates is 42, the largest 367.

\begin{table*}[pos=ht]
	\centering
	\begin{tabular}{cccccc}
                \toprule
		\# &Inconsistencies & Range (in pages) &
		\multicolumn{1}{c@{\hspace*{\tabcolsep}\makebox[0pt]{$\supset$}}}{Distances $<=10$ pages}&
		Distances $=1$ page & Max distance\\
		\midrule
		1&37 & 224\,575 & 61 & 43 & 103\,122\\
		2&30 & 423\,245 & 47 & 26 & 79\,613\\
		3&21 & 141\,591 & 20 & 5 & 54\,774\\
		4&17 & 150\,635 & 33 & 5 & 53\,319\\
		5&16 & 267\,028 & 44 & 23 & 82\,596 \\
		6&15 & 79\,296 & 85 & 42 & 71\,215\\
		7&2 & 99\,921 & 81 & 45 & 55\,761 \\
		8&0 & 82\,526 & 76 & 40 & 62\,653\\
		9&0 & 12\,132 & 75 & 57 & 3\,170\\
		10&0 & 4\,431 & 97 & 26 & 2\,665\\
                \bottomrule
        \end{tabular}
        \caption{The table shows the number of quasi\hyph instantaneous
	inconsistencies found in the ten memory dumps taken of the live system
	with high activity. For each dump the range in which the physical
	addresses of the 100 list elements and the global counter array are
	contained is given in pages. Additionally, the number of distances between the list
	elements and global counter array that are smaller than 11 pages and the
	subset of these distances that is equal to one page are shown. The last
	column shows the maximal observed distance between two list elements in
	the memory dump.}
        \label{tab:dist}
\end{table*}


\paragraph{Pivot program - pros and cons:}
	The examples show how checking for quasi\hyph instantaneous consistency allows
gaining more insights into where content mismatches occur and how many changes
on the memory regions were missed. This is currently limited to the pivot
program. But the usage of the pivot program also has benefits, since its size can be
chosen (for example by changing the number of list elements or their
size) and the degree of activity can be manipulated by the number of threads and
the frequency at which they access the list elements. As the pivot program is
started as a user process, it also allows us to gather insights into the influence
memory allocation strategies and fragmentation have on the process address space
layout at the physical level. The different ranges and distances shown in
Table~\ref{tab:dist} show that even when starting the program at approximately
the same time after booting the layout varies distinctly.
Should it be the goal
to observe the consistency when the physical pages are located in specific page
ranges, it would also be possible to remap the process's heap pages to different
physical addresses using, for example, a kernel module.

VMA inconsistencies could also be analyzed more thoroughly. For example, it would
be possible to determine the addresses of the elements of the red-black tree,
the elements of the VMA list, and the counter for VMAs, and judging from their
relative locations it would be possible to estimate where the inconsistency
could stem from, e.g., is the counter outdated or the number of elements in the
list. But identifying where exactly information is missing would require more
effort and may be impossible. It would also be more difficult or impossible to
find out how many updates were missed exactly. The observed memory range is also
limited when looking only at VMA inconsistencies. 

To cover a bigger memory range
having multiple indicators for content mismatches at hand would be convenient,
searching for them might be eased by understanding the structures used by the
operating system and their connections with each other better
\citep{pagani_back_2019}.


\section{Conclusions and Future Work}
\label{sec:conclusion}

So, finally, how can we obtain a good memory snapshot?
While this is trivial for systems that can be paused (instantaneous snapshot), the 
situation is more complex for running ones. 
We, therefore, looked into the notion of quasi\hyph instantaneous consistency which is a similar property
but also works for active systems.

In this paper, we showcased a method to observe quasi\hyph instantaneous consistency, 
validated it theoretically, and also demonstrated it in a case study.
Our method allows assessing a portion of a memory dump for quasi\hyph instantaneous
inconsistencies based on a \emph{single} memory dump. This
includes locating the memory regions
with inconsistencies and evaluating how many events on them were missed.
A property that is useful when searching for improvements in existing memory
acquisition methods. For example, identifying alternative orders of memory
acquisition, comparable to the adaptations to LiME suggested by \citet{introducing}.
Our tests then confirmed that instantaneous snapshots (frozen systems) are indeed perfect
and are the method of first resort. Furthermore, for live systems, we were able to show that a high system load 
results in more inconsistencies. More tests are needed here, but this could potentially mean
that it may be wise to close (non-relevant) applications before obtaining a snapshot 
on a running system.
%
%
%
%


Moving forward, the method could be used to evaluate different memory
acquisition tools. Thereby, a broader data base could be built to investigate our
preliminary observations further. The method for observing quasi\hyph instantaneous consistency
could also be moved to different memory ranges than the pivot program. Concerning
the main memory, it would be possible to integrate it at the hypervisor level. From
there the address ranges that contain contents of interest could be identified
and observed. For example, the memory ranges containing structures that are
necessary for the analysis, like page tables and process structures. A second consideration
is alienating this method to live acquisition of hard disk contents where it
could be possible to integrate it into the file system or block device drivers.

Our implementation also suffers some shortages which require attention. For instance, the 
pivot program utilizes a fixed-size global counter array which is not practical for larger
memory regions or longer observation times.  One possibility would be to implement
it as a ring buffer, i.e., 
restarting at the beginning once the last entry has been written. This would require an overflow
detection in the implementation, and in the analysis, with the latter being the
more difficult task.

\subsection*{Acknowledgments}

Work was supported by Deutsche
Forschungsgemeinschaft (DFG, German Research Foundation) as part of
the Research and Training Group 2475 ``Cybercrime and Forensic
Computing'' (grant number 393541319/GRK2475/1-2019).

\printcredits

\bibliography{../quellen}

\end{document}